\documentclass[11pt,english]{article}
\usepackage[T1]{fontenc}
\usepackage[latin9]{inputenc}
\usepackage{geometry}
\usepackage{verbatim}
\usepackage{float}
\usepackage{amsmath}
\usepackage{amsthm}
\usepackage{amssymb}

\makeatletter

\floatstyle{ruled}
\newfloat{algorithm}{tbp}{loa}
\providecommand{\algorithmname}{Algorithm}
\floatname{algorithm}{\protect\algorithmname}

\theoremstyle{plain}
\newtheorem*{thm*}{\protect\theoremname}
\theoremstyle{plain}
\newtheorem*{cor*}{\protect\corollaryname}
\theoremstyle{definition}
\newtheorem{defn}{\protect\definitionname}
\theoremstyle{plain}
\newtheorem{thm}{\protect\theoremname}
\theoremstyle{plain}
\newtheorem{cor}{\protect\corollaryname}
\theoremstyle{plain}
\newtheorem{lem}{\protect\lemmaname}
\theoremstyle{remark}
\newtheorem{rem}{\protect\remarkname}
\theoremstyle{remark}
\newtheorem*{acknowledgement*}{\protect\acknowledgementname}

\usepackage{complexity,color,bm}
\usepackage{colortbl}

\providecommand{\E}{\mathrm{E}}

\definecolor{gray-comment}{gray}{0.5}

\theoremstyle{plain}

\makeatletter
\newtheorem*{rep@theorem}{\rep@title}
\newcommand{\newreptheorem}[2]{%
\newenvironment{rep#1}[1]{%
 \def\rep@title{#2 \ref{##1}}%
 \begin{rep@theorem}}%
 {\end{rep@theorem}}}
\makeatother

\newreptheorem{rcor}{Corollary}
\newreptheorem{rthm}{Theorem}

\usepackage{thmtools}
\usepackage{thm-restate}

\makeatother

\usepackage{babel}
\providecommand{\acknowledgementname}{Acknowledgement}
\providecommand{\corollaryname}{Corollary}
\providecommand{\definitionname}{Definition}
\providecommand{\lemmaname}{Lemma}
\providecommand{\remarkname}{Remark}
\providecommand{\theoremname}{Theorem}

\begin{document}
\title{Beyond matroids: secretary problem and prophet inequality with general
constraints.}
\author{Aviad Rubinstein\thanks{UC Berkeley.
Most of the research was done while the author was an intern at Microsoft Research New England. Part of this research was also supported by Microsoft Research PhD Fellowship, as well as NSF grant CCF1408635 and by Templeton Foundation grant 3966. This work was done in part at the Simons Institute for the Theory of Computing.}}
\maketitle
\begin{abstract}
We study generalizations of the ``Prophet Inequality'' and ``Secretary
Problem'', where the algorithm is restricted to an arbitrary downward-closed
set system. For $\left\{ 0,1\right\} $ values, we give $O\left(\log n\right)$-competitive
algorithms for both problems. This is close to the $\Omega\left(\log n/\log\log n\right)$
lower bound due to Babaioff, Immorlica, and Kleinberg. For general
values, our results translate to $O\left(\log n\cdot\log r\right)$-competitive
algorithms, where $r$ is the cardinality of the largest feasible
set. This resolves (up to the $O\left(\log r\cdot\log\log n\right)$
factors) an open question posed to us by Bobby Kleinberg \cite{bobby-question}.
\end{abstract}

\section{Introduction}

The ``Secretary Problem'' and ``Prophet Inequality'' and their variants
have been central problems in optimal stopping theory for decades.
In both problems, an online decision maker selects one of $n$ items
that arrive online. The decision maker observes only one item at a
time, and must decide immediately and irrevocably whether to select
this item. The expected value of the item selected by the decision
maker is compared to the offline optimum. There is little that an
online algorithm can do in an arbitrary worst case, so the input is
restricted: In the Secretary Problem, the values of items are chosen
adversarially, but their arrival order is random. In Prophet Inequality,
the order is adversarial, but the values are drawn from known, independent
but not identical distributions. Algorithms with optimal competitive
ratios (with respect to the offline optimum) of $e$ and $2$ were
known since the 1960's (e.g. \cite{Dynkin63}) and 1970's \cite{KS77-prophet},
respectively.

In the past decade, variants of both problems received significant
attention from theoretical computer scientists thanks to their rich
algorithmic structure and applications to online and offline mechanism
design (e.g. \cite{BIKK08-secretary_and_AGT-survey,HKS07-prophet_and_online-MD,KW12-matroid_prophet}).
In particular, there have been many works on settings where the decision
maker is allowed to select any subset of the items subject to a certain
family of combinatorial constraints. In the famous ``Matroid Secretary
Problem'' \cite{BIK07-secretary_original}, for example, the decision
maker is allowed to select any independent set in a given matroid.
Obtaining a constant competitive ratio algorithm is a long-standing
open problem: the state of the art is $O\left(\log\log r\right)$,
where $r$ is the rank of the matroid \cite{Lachish14-secretary,FSZ15-loglogr}.
Its cousin, the ``Matroid Prophet Inequality'', has been resolved
by Kleinberg and Weinberg who gave a $2$-competitive algorithm \cite{KW12-matroid_prophet}.
Other constraints such as special classes of matroids (\cite{Dinitz13-survey}
and references therein), polymatroids \cite{DK15-polymatroid}, knapsack
and matchings \cite{FSZ16-OCRS}, etc. have also been studied in these
settings. Even the 2013 survey by Dinitz \cite{Dinitz13-survey} seems
outdated with so much exciting progress in the last couple of years
(e.g. \cite{AKW14-prophet_limited_info,FZ15-submodular_secretary,KKN15-limited_randomness,FGKN15-temp_secretary,Vardi15-returning_secretary}).

\subsubsection*{Our results}

In this work we revisit an important missing piece in this puzzle:
what happens when there are no guarantees on the family of feasible
sets? Babaioff, Immorlica, and Kleinberg \cite{BIK07-secretary_original}
gave an $\Omega\left(\log n/\log\log n\right)$ lower bound\footnote{Stated in terms of the $r$, the maximum size of a feasible set, \cite{BIK07-secretary_original}'s
construction gives an $\Omega\left(r\right)$-lower bound on the competitive
ratio; one can easily obtain a matching $O\left(r\right)$-upper bound
by applying the algorithms for the classical (single item) variants
of both problems.} on the competitive ratio of the Secretary Problem with arbitrary
downward-closed constraints, even in the special case where all values
are in $\left\{ 0,1\right\} $. (Essentially the same construction
gives the same lower bound for the corresponding Prophet Inequality;
see also Appendix \ref{sec:lower_bound}.) Beyond this lower bound,
the problem remained poorly understood. In particular there was a
disturbing lack of upper bounds on the competitive ratio - nothing
beyond the trivial $O\left(n\right)$ was known. Here, we make significant
progress towards closing this gap:
\begin{thm*}
[Main Theorem] When the items take values in $\left\{ 0,1\right\} $,
the competitive ratios for {\sc Downward-Closed Prophet} and {\sc Downward-Closed Secretary}
are $O\left(\log n\right)$.
\end{thm*}
For general (non-negative) valuations, we can reduce to the $\left\{ 0,1\right\} $-valued
case with a loss of $O\left(\log r\right)$, where $r$ denotes the
cardinality of the largest feasible set.
\begin{cor*}
For general item values, the competitive ratios for {\sc Downward-closed Prophet}
and {\sc Downward-Closed Secretary} are $O\left(\log n\cdot\log r\right)$.
\end{cor*}

\paragraph*{A recent observation by Kesselheim, Singla, and Svensson}

It was already noted in the first version of this paper that the Secretary
Problem algorithm requires very minimal assumptions on the order of
arrival of secretaries (see Remark \ref{rem:order-oblivious}). It
was recently brought to my attention\footnote{July 2024 private communication with Sahil Singla, based on a Fall
2022 discussion between Thomas Kesselheim, Sahil Singla, and Ola Svensson.} that this implies that the secretary algorithm satisfies the {\em order-oblivious}
criterion of \cite{AKW14-prophet_limited_info}. Therefore, by a black-box
reduction of \cite{AKW14-prophet_limited_info} it immediately implies
our results for {\sc Downward-Closed Prophet}, with a completely
different algorithm and analysis. Furthermore, it implies that the
results for {\sc Downward-Closed Prophet} can be obtained with a
{\em single-sample} algorithm, i.e.~an algorithm that, for each
distribution, takes as input only a single sample instead of a full
description of the distribution! 

\paragraph{Computational Efficiency}

Our algorithm for {\sc Downward-Closed Secretary} assumes access
to a demand-like oracle. The algorithm for {\sc Downward-Closed Prophet}
makes slightly more complicated queries of the form: ``conditioning
on the history of realizations and items selected by the algorithm
so far, what is the probability that the prophet can still obtain
value $\tau$?''. Our algorithms are efficient assuming access to
such oracles. Note that it is unreasonable to expect much more: optimizing
over arbitrary families of downward-closed sets is computationally
intractable even in the offline case (for example, when the feasible
sets are the independent sets of a graph \cite{bobby-question}).

\paragraph{Non-monotone feasibility constraints}

One can further generalize the problem and consider feasibility constraints
which are not even downward-closed. In Section \ref{sec:Non-monotone-feasibility-constra}
we briefly discuss {\sc Non-monotone Prophet} and {\sc Non-monotone Secretary}
and prove near tight lower bounds on the competitive ratios. 

\subsection{Techniques}

Both algorithms rely on potential function argument. The key idea
for the analysis of {\sc Downward-closed Prophet} is the use of a
{\em dynamic potential function}: when the algorithm cannot guarantee
adequate progress with respect to the current potential function,
it dynamically modifies the potential function. We believe that this
simple idea will find applications in other potential function arguments.

We note that our algorithms for {\sc Downward-closed Prophet} and
{\sc Downward-Closed Secretary} bear little resemblance to each other
or to related works on the Matroid Secretary Problem and Matroid Prophet
Inequality. In particular, we focus on the case of $\left\{ 0,1\right\} $-valued
items, which is easy for matroids. Defying our intuition from matroids
again, our algorithm and analysis for {\sc Downward-closed Secretary}
happens to be simpler than that of {\sc Downward-Closed Prophet}.
Nevertheless, we begin with {\sc Downward-Closed Prophet} in the
next section, as we find the techniques more exciting. The algorithm
for {\sc Downward-closed Secretary} is described in Section \ref{sec:Secretary}
and non-monotone constraints are discussed in Section \ref{sec:Non-monotone-feasibility-constra}.

\section{\label{sec:Prophet}Prophet}
\begin{defn}
[{\sc Downward-Closed Prophet}] Consider $n$ items with independently
distributed values $\left\{ X_{i}\sim{\cal D}_{i}\right\} _{i=1}^{n}$.
Let ${\cal F}$ be an arbitrary downward-closed set system over the
items $\left[n\right]$, and initialize $W$ as the empty set. The
algorithm receives as input $n$, ${\cal F}$, and the distributions,
and observes the realizations (the value of $X_{i}$) online. After
observing the realization of $X_{i}$, the algorithm must decide (immediately
and irrevocably) whether to add $i$ to the set $W$, subject to the
constraint that $W$ remains a feasible set in ${\cal F}$. The objective
is to maximize the sum of values of items in $W$. 
\end{defn}
\begin{thm}
\label{thm:prophet}When the $X_{i}$'s take values in $\left\{ 0,1\right\} $,
there is a deterministic algorithm for {\sc Downward-Closed Prophet}
that achieves a competitive ratio of $O\left(\log n\right)$.
\end{thm}
Let $r$ denote the maximum cardinality of a feasible set $S\in{\cal F}$. 
\begin{cor}
\label{cor:prophet}There is a deterministic algorithm for {\sc Downward-Closed Prophet}
that achieves a competitive ratio of $O\left(\log n\cdot\log r\right)$.
\end{cor}
\begin{proof}
[Proof of Corollary \ref{cor:prophet} from Theorem \ref{thm:prophet}]
We recover separately the contributions from ``tail'' events (a single
item taking an exceptionally high value) and the ``core'' contribution
that is spread over many items. Run the better of the following two
algorithms:

\paragraph*{Tail }

Let $OPT$ denote the expected offline optimum value. Whenever we
see an item with value at least $2OPT$, we select it. For item $i$,
let $p_{i}=\Pr\left[X_{i}\geq2OPT\right]$. We have 
\[
OPT\geq2OPT\cdot\Pr\left[\exists i\colon X_{i}\geq2OPT\right]=2OPT\cdot\left(1-\prod\left(1-p_{i}\right)\right).
\]
Dividing by $OPT$ and rearranging, we get
\[
1/2\leq\prod\left(1-p_{i}\right)\leq e^{-\sum p_{i}},
\]
and thus 
\[
\sum p_{i}\leq\ln2.
\]

Therefore the probability that we want to take an item but can't is
at most $\ln2$, so this algorithm achieves at least a $\left(1-\ln2\right)$-fraction
of the expected contribution from values greater than $2OPT$.

\paragraph*{Core }

Observe that we can safely ignore values less than $OPT/2r$, as those
can contribute a total of at most $OPT/2$. Partition all remaining
values into $\log r+2$ intervals $\left[OPT/2r,OPT/r\right],\allowbreak\dots,\mbox{\allowbreak}\left[OPT,2OPT\right]$.
The expected contribution from the values in each interval is a $\Omega\left(1/\log r\right)$-fraction
of the expected offline optimum without values greater than $2OPT$.
Pick the interval with the largest expected contribution, round down
all the values in this interval, and run the algorithm guaranteed
by Theorem \ref{thm:prophet}. This achieves an $\Omega\left(\frac{1}{\log n\cdot\log r}\right)$-fraction
of the expected contribution from values less than or equal to $2OPT$.
\end{proof}
The rest of this section is devoted to the proof of Theorem \ref{thm:prophet}.

\subsection{A dynamic potential function}

A natural approach to solving {\sc Downward-Closed Prophet} is the
following meta-algorithm: wait for an item with value $1$, and select
it if that does not decrease some potential function. In Appendix
\ref{sec:Naive-approaches} we informally discuss a few intuitive
potential functions and the difficulties that arise in analyzing each
of them. Here we go directly to the remedy to those obstacles: a ``dynamic''
potential function.

The basic question we want to ask about a new item is the following:
given the realizations we have observed so far and items we have already
selected, is it a good idea to select this item? The main challenge
is to come up with an analyzable proxy to ``good idea''. To this end,
we consider a restricted prophet who is committed to past decisions
-and must also select the current item- but is omniscient about future
realizations; we maintain a target value $\tau$, and ask what is
the probability over future realizations that the restricted prophet
can obtain a solution with value $\tau$. The main novelty in our
analysis is that we dynamically update $\tau$. In particular, as
the prophet becomes more restricted by accumulating commitments, we
compare his solution's value to a lower target $\tau$. 

At each iteration, the algorithm maintains a target value $\tau$
and a target probability $\pi$; $\pi$ is the probability (over future
realizations) that the current restricted prophet beats $\tau$. We
say that an item is {\em good} if selecting it does not decrease
the probability of beating the target value by a factor greater than
$n^{2}$, and {\em bad} otherwise. Notice that all the bad items
together contribute at most a $\left(1/n\right)$-fraction of the
probability of beating $\tau$. As we mentioned above, a key ingredient
is that $\tau$ is updated dynamically. If the probability of observing
a good item with value $1$ is too low (less than $1/3$), we deduct
$1$ from $\tau$. We show (Lemma \ref{lem:main}) that this increases
$\pi$ by a factor of at least $2$. $\pi$ decreases by an $n^{2}$
factor when we select an item, and increases by a factor of $2$ whenever
we deduct $1$ from $\tau$: we balance $2\log n$ deductions for
every item the algorithm selects, and this gives the $O\left(\log n\right)$
competitive ratio. 

So far our algorithm is roughly as follows: set a target value $\tau$;
whenever the probability $\pi$ of reaching the target $\tau$ drops
below $1/3$, decrease $\tau$; if $\pi>1/3$, sit and wait for a
good item with value $1$ - one will arrive with probability at least
$1/3-o\left(1\right)$. There is one more subtlety: what should the
algorithm do if all the good items have value $0$? In other words,
what if the probability of observing a good item with value $1$ is
neither very low nor very close to $1$, say $1/2$ or even $1-\frac{1}{\log n}$?
On one hand, we can't decrease $\tau$ again, because we are no longer
guaranteed a significant increase in $\pi$; on the other hand, after,
say $\Theta\left(\log^{2}n\right)$ iterations, we still have a high
probability of having an iteration where none of the good items has
value $1$. (If no good $1$'s are coming, we don't want the algorithm
to wait forever...) Fortunately, there is a simple solution: the algorithm
waits for the last good item; if, against the odds, no $1$'s have
yet been observed, the algorithm ``hallucinates'' that this last item
has value $1$, and selects it. In expectation, at most a constant
fraction of the items we select will have value $0$, so the competitive
ratio is still $O\left(\log n\right)$.

\subsection{Formal description of the algorithm}

\subsubsection*{Notation}

We let $OPT$ denote the expected (offline) optimum. $W$ is the set
of items selected so far ($W$ for ``Wins''), and $\ell_{W}\triangleq\max\left\{ i\in W\right\} $
is the index of the last selected item.

Let ${\cal F}$ denote the family of all feasible sets. For any $T\subseteq\left[n\right]$,
let ${\cal F}_{T}$ denote the family of feasible sets whose intersection
with $\left\{ 1,\dots,\max\left\{ T\right\} \right\} $ is exactly
$T$. 

Let $X_{i}$ denote the random value of the $i$-th item. We use $z_{i}$
to refer to the observed realization of $X_{i}$. We let $V\left({\cal F},X_{\left[n\right]}\right)\triangleq\max_{S\in{\cal F}}\sum_{i\in S}X_{i}$
denote the value of optimum offline solution (note that this is also
a random variable).

Let $\tau=\tau\left(W\right)$ be the current target value, and $\pi=\pi\left(\tau,W\right)$
denotes the target probability:
\[
\pi\left(\tau,W\right)\triangleq\Pr\left[V\left({\cal F}_{W},X_{\left[n\right]}\right)>\tau\mid X_{\left[\ell_{W}\right]}=z_{\left[\ell_{W}\right]}\right].
\]
We also define $\pi^{j}=\pi^{j}\left(\tau,W\right)$ to be the probability
of reaching $\tau$, given that $j$ is the next item we select. Formally,
\[
\pi^{j}\left(\tau,W\right)\triangleq\Pr\left[V\left({\cal F}_{W\cup\left\{ j\right\} },X_{\left[n\right]}\right)>\tau\mid X_{\left[j\right]}=\left(z_{\left[\ell_{W}\right]},\underbrace{0\ldots0}_{\ell_{_{W}}+1,\dots,j-1},1\right)\right].
\]

We say that a future item is {\em good} (and {\em bad} otherwise)
if $\pi^{j}\geq n^{-2}\cdot\pi$. Finally, $G=G\left(\tau,\pi,W\right)\triangleq\left\{ j>\ell_{W}\colon\pi^{j}\geq n^{-2}\cdot\pi\right\} \cap\left(\bigcup_{S\in{\cal F}_{W}}S\right)$
is the set of items that are both good and feasible, and $A=A\left(G\right)\triangleq\left(\bigvee_{j\in G}X_{j}=1\right)$
is the event that at least one of the good items has value $1$.

(See also list of symbols in Appendix \ref{sec:List-of-symbols}.)

\subsubsection*{Algorithm}

Initialize $\tau\leftarrow OPT/2$ and $W\leftarrow\emptyset$. 

After each update to $W$, decrease $\tau$ until $\Pr\left[A\right]\geq1/3$,
or until $\left|W\right|>\tau$. When $\Pr\left[A\right]\geq1/3$,
reveal the values of items until observing a good item (i.e. some
$j\in G$) with value $1$. When we observe a good item with value
$1$, add it to $W$. If we reached the last good item without observing
any good items with value $1$, add the last good item to $G$ and
subtract $1$ from $\tau$%
. See also pseudocode in Algorithm \ref{alg:prophet}.

\begin{algorithm}
\caption{\label{alg:prophet}Prophet}

\begin{enumerate}
\item $\tau\leftarrow\frac{OPT}{2}$; $W\leftarrow\emptyset$
\item while $\tau>\left|W\right|$:

\begin{enumerate}
\item $\pi\leftarrow\Pr\left[V\left({\cal F}_{W},X_{\left[n\right]}\right)>\tau\mid X_{\left[\ell_{W}\right]}=z_{\left[\ell_{W}\right]}\right]$

{\color{gray-comment} \# $\pi$ is the probability that, given the
history, the offline optimum can still beat $\tau$.}
\item $G\leftarrow\left\{ j>\ell_{W}\colon\pi^{j}\geq n^{-2}\cdot\pi\right\} \cap\left(\bigcup_{S\in{\cal F}_{W}}S\right)$

{\color{gray-comment} \# $G$ is the set of good and feasible items.}
\item if $\Pr\left[A\right]\geq1/3$

{\color{gray-comment} \# A good item with value $1$ is likely to
arrive.}
\begin{enumerate}
\item $j^{*}\leftarrow\min\left\{ j\in G\colon z_{j}=1\right\} $

{\color{gray-comment} \# Wait for a good and feasible item with value
$1$.}
\item if $j^{*}=\infty$

{\color{gray-comment} \# All good items have value $0$.}
\begin{enumerate}
\item $j^{*}\leftarrow\max G$

{\color{gray-comment} \# Select the last item anyway.}
\item $\tau\leftarrow\tau-1$

{\color{gray-comment} \# Adjust the target value to account for select
an item with value $0$}
\end{enumerate}
\item $W\leftarrow W\cup\left\{ j^{*}\right\} $
\end{enumerate}
\item else

\begin{enumerate}
\item $\tau\leftarrow\tau-1$\label{enu:deduct-from-tau}

{\color{gray-comment} \# decrease target value $\tau$ until $\Pr\left[A\right]\geq1/3.$}
\end{enumerate}
\end{enumerate}
\end{enumerate}
\end{algorithm}

\subsection{Analysis}

\subsection*{Concentration}

We want to argue that the value of the optimum concentrates around
its expectation. Proving concentration for a maximum over an arbitrary
family of sets' sums is rather non-trivial. Fortunately, there is
a vast literature on concentration bounds for suprema of empirical
processes. We use the following inequality due to Ledoux. (It is particularly
convenient because the denominator in the exponent depends on the
expected supremum rather on absolute bounds on the values each item
can take.)
\begin{thm}
{\cite[Theorem 2.4]{Ledoux1997}}\label{thm:ledoux} There exists
some constant $K>0$ such that the following holds. Let $Y_{i}$'s
be independent (but not necessarily identical) random variables in
some space $S$; let ${\cal C}$ be a countable class of measurable
functions $f\colon S\rightarrow\left[0,1\right]$; and let $Z=\sup_{f\in{\cal C}}\sum_{i=1}^{n}f\left(Y_{i}\right)$.
Then, 
\[
\Pr\left[Z\geq\E\left[Z\right]+t\right]\le\exp\left(-\frac{t}{K}\cdot\log\left(1+\frac{t}{\E\left[Z\right]}\right)\right).
\]
\end{thm}
To make the connection to our setting, let $Y_{i}$ be the vector
in $\left[0,1\right]^{\left|{\cal F}\right|}$ whose $S$-th coordinate
is $X_{i}$ if $i\in S$, and $0$ otherwise. Let $f_{S}\left(Y_{i}\right)\triangleq\left[Y_{i}\right]_{S}$,
so $\sum_{i=1}^{n}f_{S}\left(Y_{i}\right)$ is simply the value of
set $S$. Let ${\cal C}\triangleq\left\{ f_{S}\right\} _{S\in{\cal F}}$.
The above concentration inequality can now be written as
\begin{equation}
\Pr\left[V\left({\cal F},X_{\left[n\right]}\right)\geq OPT+t\right]\le\exp\left(-\frac{t}{K}\cdot\log\left(1+\frac{t}{OPT}\right)\right).\label{eq:V-concentrates}
\end{equation}

In fact, we only need the following much weaker lemma. Notice that
we can assume without loss of generality that $OPT\geq\log n$; otherwise
the trivial greedy algorithm guarantees an expected value of $\Omega\left(\min\left\{ OPT,1\right\} \right)=\Omega\left(OPT/\log n\right)$.
\begin{lem}
\label{lem:concentration}Assume $OPT\geq\Omega\left(\log n\right)$.
Then,
\[
\Pr\left[V\left({\cal F},X_{\left[n\right]}\right)\geq\frac{OPT}{2}\right]>1/4.
\]
\end{lem}
\begin{proof}
We have, 
\begin{align}
OPT & =\int_{-OPT}^{\infty}\Pr\left[V\left({\cal F},X_{\left[n\right]}\right)\geq OPT+t\right]dt,\label{eq:OPT-integral}
\end{align}
which can be decomposed as to integrals over $\left[-OPT,-OPT/2\right]$,
$\left[-OPT/2,OPT\right]$, and $\left[OPT,\infty\right]$. 

The first two integrals can be easily bounded as
\[
\int_{-OPT}^{-OPT/2}\Pr\left[V\left({\cal F},X_{\left[n\right]}\right)\geq OPT+t\right]dt\leq\int_{-OPT}^{-OPT/2}1\cdot dt\leq\frac{OPT}{2}
\]
and
\begin{eqnarray*}
\int_{-OPT/2}^{OPT}\Pr\left[V\left({\cal F},X_{\left[n\right]}\right)\geq OPT+t\right]dt & \leq & \int_{-OPT/2}^{OPT}\Pr\left[V\left({\cal F},X_{\left[n\right]}\right)\geq OPT/2\right]dt\\
 & \leq & \frac{3OPT}{2}\cdot\Pr\left[V\left({\cal F},X_{\left[n\right]}\right)>\frac{OPT}{2}\right].
\end{eqnarray*}

For the third integral we use the concentration bound (\ref{eq:V-concentrates}):
\begin{eqnarray*}
\int_{OPT}^{\infty}\Pr\left[V\left({\cal F},X_{\left[n\right]}\right)\geq OPT+t\right]dt & \leq & \int_{OPT}^{\infty}\exp\left(-\frac{t}{K}\cdot\log\left(1+\frac{t}{OPT}\right)\right)dt\\
 & \leq & \int_{OPT}^{\infty}\exp\left(-\frac{t}{K}\right)dt\\
 & = & \left[Ke^{-t/K}\right]_{OPT}^{\infty}=K\cdot e^{-OPT/K},
\end{eqnarray*}
which is negligible since $OPT=\omega\left(1\right)$. 

Plugging into (\ref{eq:OPT-integral}), we have:
\[
OPT\leq\frac{OPT}{2}+\frac{3OPT}{2}\cdot\Pr\left[V\left({\cal F},X_{\left[n\right]}\right)>\frac{OPT}{2}\right]+o\left(1\right),
\]
and after rearranging we get
\[
\Pr\left[V\left({\cal F},X_{\left[n\right]}\right)>\frac{OPT}{2}\right]\geq1/3-o\left(1\right).
\]
\end{proof}

\subsection*{Main lemma}
\begin{lem}
\label{lem:main}At any point during the run of the algorithm, if
$\Pr\left[A\right]\leq1/3$, then subtracting $1$ from $\tau$ doubles
$\pi$; i.e. 
\[
\Pr\left[V\left({\cal F}_{W},X_{\left[n\right]}\right)>\tau-1\mid X_{\left[\ell_{W}\right]}=z_{\left[\ell_{W}\right]}\right]\geq2\Pr\left[V\left({\cal F}_{W},X_{\left[n\right]}\right)>\tau\mid X_{\left[\ell_{W}\right]}=z_{\left[\ell_{W}\right]}\right].
\]
\end{lem}
\begin{proof}
We first claim that the probability that the optimum solution (conditioned
on the items $W$ we already selected and the realizations $z_{\left[\ell_{W}\right]}$
we have already seen) reaches $\tau$ without selecting any items
from $G$ is negligible with respect to the total probability of reaching
$\tau$. For each $k\notin G$, we have, by definition of $G$, 
\[
\underbrace{\Pr\left[V\left({\cal F}_{W\cup\left\{ k\right\} },X_{\left[n\right]}\right)>\tau\mid X_{\left[k\right]}=\left(z_{\left[\ell_{W}\right]},0\dots0,1\right)\right]}_{\pi^{k}}<n^{-2}\cdot\underbrace{\Pr\left[V\left({\cal F}_{W},X_{\left[n\right]}\right)>\tau\mid X_{\left[\ell_{W}\right]}=z_{\left[\ell_{W}\right]}\right]}_{\pi}.
\]
However, if $X_{j}=0$ for all $j\in G$ (i.e. if event $A$ does
not happen), the first item that contributes to the optimum solution
has to be some $k\notin G$. Thus the probability that it still reaches
$\tau$ is very small: 
\begin{flalign*}
\Pr\left[\left(V\left({\cal F}_{W},X_{\left[n\right]}\right)>\tau\right)\wedge\left(\neg A\right)\mid\left(X_{\left[\ell_{W}\right]}=z_{\left[\ell_{W}\right]}\right)\right] & \leq\sum_{k\notin G}\Pr\left[\left({\cal F}_{W\cup\left\{ k\right\} },X_{\left[n\right]}\right)>\tau\mid X_{\left[k\right]}=\left(z_{\left[\ell_{W}\right]},0\dots0,1\right)\right]\\
 & \leq n^{-1}\cdot\Pr\left[V\left({\cal F}_{W},X_{\left[n\right]}\right)>\tau\mid X_{\left[\ell_{W}\right]}=z_{\left[\ell_{W}\right]}\right].
\end{flalign*}
In particular, most of the probability of reaching $\tau$ comes from
the event $A$ (i.e. $X_{j}=1$ for some $j\in G$):
\[
\Pr\left[\left(V\left({\cal F}_{W},X_{\left[n\right]}\right)>\tau\right)\wedge A\mid X_{\left[\ell_{W}\right]}=z_{\left[\ell_{W}\right]}\right]\geq\left(1-n^{-1}\right)\Pr\left[V\left({\cal F}_{W},X_{\left[n\right]}\right)>\tau\mid X_{\left[\ell_{W}\right]}=z_{\left[\ell_{W}\right]}\right].
\]

Inconveniently, this does not guarantee that $\Pr\left[A\right]$
is large, because the probability of reaching $\tau$ may be very
small to begin with, and those events are not independent. Instead,
rewrite the LHS of the above equation as
\[
\Pr\left[\left(V\left({\cal F}_{W},X_{\left[n\right]}\right)>\tau\right)\wedge A\mid X_{\left[\ell_{W}\right]}=z_{\left[\ell_{W}\right]}\right]=\underbrace{\Pr\left[A\right]}_{\leq1/3}\Pr\left[V\left({\cal F}_{W},X_{\left[n\right]}\right)>\tau\mid\left(X_{\left[\ell_{W}\right]}=z_{\left[\ell_{W}\right]}\right)\wedge A\right]
\]
 By the premise, this implies a strong lower bound on the second term
on the RHS. Finally, we claim that conditioning on $A$ can contribute
at most $1$ to the optimum $V\left({\cal F}_{w},X_{\left[n\right]}\right)$.
Assuming this claim, the proof of the lemma is complete:
\begin{align}
\hspace{-0.1cm}\hspace{-0.1cm}\Pr\left[V\left({\cal F}_{W},X_{\left[n\right]}\right)>\tau-1\mid X_{\left[\ell_{W}\right]}=z_{\left[\ell_{W}\right]}\right] & \hspace{-0.1cm}\geq\Pr\left[V\left({\cal F}_{W},X_{\left[n\right]}\right)>\tau\mid\left(X_{\left[\ell_{W}\right]}=z_{\left[\ell_{W}\right]}\right)\wedge A\right]\label{eq:claim}\\
 & \hspace{-0.1cm}\geq\left(3-o\left(1\right)\right)\Pr\left[V\left({\cal F}_{W},X_{\left[n\right]}\right)>\tau\mid X_{\left[\ell_{W}\right]}=z_{\left[\ell_{W}\right]}\right].\nonumber 
\end{align}

It remains to prove (\ref{eq:claim}). Rewrite $A$ as a disjoint
union of the events $A^{j}\triangleq\left(X_{j}=1\bigwedge_{\stackrel{i<j}{i\in G}}X_{i}=0\right)$,
i.e. $A^{j}$ is the event that $j$ is the first good item with value
$1$. Let $\pi_{A}$ denote the probability of beating the target
value conditioning on the history and event $A$, and analogously
for $\pi_{A^{j}}$:
\begin{eqnarray*}
\pi_{A} & \triangleq & \Pr\left[V\left({\cal F}_{W},X_{\left[n\right]}\right)>\tau\mid\left(X_{\left[\ell_{W}\right]}=z_{\left[\ell_{W}\right]}\right)\wedge A\right]\\
\pi_{A^{j}} & \triangleq & \Pr\left[V\left({\cal F}_{W},X_{\left[n\right]}\right)>\tau\mid\left(X_{\left[\ell_{W}\right]}=z_{\left[\ell_{W}\right]}\right)\wedge A^{j}\right].
\end{eqnarray*}

Let $B^{j}\triangleq\left(X_{j}=0\bigwedge_{\stackrel{i<j}{i\in G}}X_{i}=0\right)$
be the event that neither the $j$-th item nor any of the preceding
good items have value $1$. Changing the value of $X_{j}$ can decrease
the final value of the solution by at most $1$. Thus, for every $j$,
\begin{eqnarray*}
\pi_{A^{j}} & \leq & \Pr\left[V\left({\cal F}_{W},X_{\left[n\right]}\right)>\tau-1\mid\left(X_{\left[\ell_{W}\right]}=z_{\left[\ell_{W}\right]}\right)\wedge B^{j}\right]\\
 & \leq & \Pr\left[V\left({\cal F}_{W},X_{\left[n\right]}\right)>\tau-1\mid X_{\left[\ell_{W}\right]}=z_{\left[\ell_{W}\right]}\right],
\end{eqnarray*}
where the last inequality follows because conditioning on $B^{j}$
can only decrease the probability of reaching the target value.

Finally, since $A$ is a disjoint union of the $A^{j}$'s, it follows
that $\pi_{A}$ is a convex combination of the $\pi_{A^{j}}$'s. Therefore,
\[
\Pr\left[V\left({\cal F}_{W},X_{\left[n\right]}\right)>\tau\mid X_{\left[\ell_{W}\right]}=z_{\left[\ell_{W}\right]}\wedge A\right]=\pi_{A}\leq\Pr\left[V\left({\cal F}_{W},X_{\left[n\right]}\right)>\tau-1\mid X_{\left[\ell_{W}\right]}=z_{\left[\ell_{W}\right]}\right].
\]
\end{proof}

\subsection*{Putting it all together}
\begin{lem}
\label{lem:potential}At any point during the run of the algorithm,
\[
\tau\geq\frac{OPT}{2}-\left(2\log n+1\right)\cdot\left|W\right|-2
\]
\end{lem}
\begin{proof}
We prove by induction that at any point during the run of the algorithm,
\begin{equation}
\log\pi\geq-2-\left(2\log n+1\right)\cdot\left|W\right|+\left(\frac{OPT}{2}-\tau\right).\label{eq:induction}
\end{equation}
After initialization, $\log\pi\geq-2$ by Lemma \ref{lem:concentration}.
By definition of $G$, whenever we add an item to $W$, we decrease
$\log\pi$ by at most $2\log n$ - hence the $2\log n\cdot\left|W\right|$
term. Notice that when the algorithm ``hallucinates'' a $1$, we also
decrease $\tau$ by $1$ to correct for the hallucination - at any
point during the run of the algorithm, this has happened at most $\left|W\right|$
times. Recall that we may also decrease $\tau$ in the last line of
Algorithm \ref{alg:prophet} (in order to increase $\pi$); whenever
we do this, $\tau$ decreases by $1$, but $\pi$ doubles (by Lemma
\ref{lem:main}), so $\log\pi$ increases by $1$, and Inequality
(\ref{eq:induction}) is preserved. 

Finally, since $\pi$ is a probability, we always maintain $\log\pi\leq0$.
\end{proof}
We are now ready to complete the proof of Theorem \ref{thm:prophet}.
\begin{proof}
[Proof of Theorem \ref{thm:prophet}]The algorithm always terminates
after at most $O\left(OPT\right)$ decreases to the value of $\tau$.
By Lemma \ref{lem:potential}, when the algorithm terminates, we have
$\left|W\right|\geq\tau\geq\frac{OPT}{2}-\left(2\log n+1\right)\cdot\left|W\right|-2$,
and therefore in particular $\left|W\right|\geq\frac{OPT-4}{4\log n+4}$. 

Finally, recall that sometimes the algorithm ``hallucinates'' good
realizations, i.e. for some items $i\in W$ that we select, $X_{i}=0$.
However, each time we add an item, the probability that we add a zero-value
item is at most $2/3$ (by the condition $\Pr\left[A\right]>1/3$).
Therefore in expectation the value of the algorithm is at least $\left|W\right|/3$.
\end{proof}

\section{\label{sec:Secretary}Secretary}

We now discuss our algorithm for the {\sc Downward-Closed Secretary}
problem. Recall that as we mentioned earlier, it is very different
from our algorithm for {\sc Downward-Closed Prophet}.
\begin{defn}
[{\sc Downward-Closed Secretary}] Consider $n$ items with values
$\left\{ y_{i}\right\} $ and an arbitrary downward-closed set system
${\cal F}$ over the items; both $\left\{ y_{i}\right\} $ and ${\cal F}$
are adversarially chosen. The algorithm receives as input $n$ (but
not ${\cal F}$ or $\left\{ y_{i}\right\} $). The items arrive in
a uniformly random order. Initialize $W$ as the empty set. When item
$i$ arrives, the algorithm observes its value and all feasible sets
it forms with items that have previously arrived. The algorithm then
decides (immediately and irrevocably) whether to add $i$ to the set
$W$, subject to the constraint that $W$ remains a feasible set in
${\cal F}$. The goal is to maximize the sum of values of items in
$W$. 
\end{defn}
\begin{rem}
\label{rem:order-oblivious}We remark that our analysis for the $\left\{ 0,1\right\} $-valued
case does not require a uniformly random arrival order. Rather, partition
the time intervals into pairs $\left\{ j,n/2+j\right\} $; the items
can be assigned arbitrarily to pairs of arriving times, and we only
require that the choice of which of the two items arrives at time
$j$ and which at time $n/2+j$ is random. 
\end{rem}
\begin{thm}
\label{thm:secretary}When $y_{i}\in\left\{ 0,1\right\} $ for all
$i$, there is a deterministic algorithm for {\sc Downward-Closed Secretary}
that achieves a competitive ratio of $O\left(\log n\right)$.
\end{thm}
Let $r$ denote the maximum cardinality of a feasible set $S\in{\cal F}$. 
\begin{cor}
\label{cor:secretary}For general valuations, there is a randomized
algorithm for {\sc Downward-Closed Secretary} that achieves a competitive
ratio of $O\left(\log n\cdot\log r\right)$.
\end{cor}
\begin{proof}
[Proof of Corollary \ref{cor:secretary} from Theorem \ref{thm:secretary}]
Run the classic secretary algorithm over the first $n/2$ items. With
constant probability the algorithm selects the item with the largest
value, which we denote by $M$. 

Also, with constant probability the algorithm sees the item with the
largest value too early and does not select it. Assume that this is
the case. Since we obtained expected value of $\Omega\left(M\right)$
on the first $n/2$ items we can, without loss of generality, ignore
values less than $M/r$. Partition all remaining values into $\log r$
interval-buckets $\left[M/r,2M/r\right],\dots,\left[M/2,M\right]$;
in expectation, each contributes a $\left(1/\log r\right)$-fraction
of the optimum. Choose a bucket at random, round all the values in
it to $1$ and set the rest to $0$. Finally, run the $O\left(\log n\right)$-competitive
algorithm from Theorem \ref{thm:secretary} on the last $n/2$ items.
(We remark that the choice of bucket can be derandomized using the
random order of the first $n/2$ items.) 
\end{proof}
The rest of this section is devoted to the proof of Theorem \ref{thm:secretary}. 

\subsection{Secretary algorithm}

We run a greedy algorithm over a sliding window of size $n/2$. At
first, we have a fully offline solution (none of which we actually
select) over the first $n/2$ items. We gradually transform it into
a real online solution over the last $n/2$ items. 

Intuitively, we test how well each item interacts with the previous
$n/2$ items to estimate how well it would interact with future items.
The main challenge is that we reuse the same randomness from the previous
items over and over, and we may overfit our choices to the past. We
overcome this by showing that the probability that the previous items
do not represent the future is very small - so small that we can take
a union bound over all the potential actions of our algorithm. Roughly
speaking, our main argument only uses $1$ bit of randomness for each
item in the optimal offline solution (does it arrive in one of the
first $n/2$ time periods?), so we cannot expect any concentration
inequality to bound the deviations to probability less than $2^{-OPT}$.
On the other hand, if our algorithm selects $\tau$ items, we have
to take a union bound over ${n \choose \tau}$ sets, so we want $2^{-OPT}\cdot{n \choose \tau}\ll1$.
Taking $\tau\approx OPT/\log n$ gives the promised competitive ratio.%

\paragraph{Notation}

For any set $T$, we let $V\left({\cal F};T\right)$ denote the value
of the offline optimum solution restricted to $T$. We let ${\cal F}_{W}$
denote the feasible sets that contain $W$, and finally $V\left({\cal F}_{W};T\right)$
denotes the value of the offline optimum among feasible sets that
contain $W$ and are contained in $T$. We use $\sigma$ to denote
the order of arrivals; in particular, $\sigma_{j}$ is the item that
arrives in the $j$-th time period.

\paragraph{Algorithm}

In the exploration phase, we observe the first $n/2$ items and add
none of them to $W$; we initialize $U$ as the set of those items
(notice that $U$ needs not be a feasible set), and set our target
value to $\tau\triangleq V\left({\cal F};U\right)/500\log n$. In
the exploitation phase, before observing the $\left(n/2+j\right)$-th
item, we ``forget'' the $j$-th item, i.e. we permanently remove $\sigma_{j}$
from $U$. Then, we add the $\sigma_{n/2+j}$ to $W$ iff $V\left({\cal F}_{W\cup\left\{ \sigma_{n/2+j}\right\} };U\cup W\cup\left\{ \sigma_{n/2+j}\right\} \right)>V\left({\cal F}_{W};U\cup W\right)$.
Continue until $\left|W\right|=\tau$ or until all items have been
revealed. (Of course, aborting when $\left|W\right|=\tau$ can only
hurt the expected value of the solution, but this simplifies the analysis
by restricting the loss from taking a union bound.) See also pseudocode
in Algorithm \ref{alg:secretary}.

\begin{algorithm}
\caption{\label{alg:secretary}Secretary}

\begin{enumerate}
\item $W\leftarrow\emptyset$; $U\leftarrow\emptyset$

{\color{gray-comment} \# Exploration phase:}
\item for $j\in\left\{ 1,\dots,n/2\right\} $:

\begin{enumerate}
\item $U\leftarrow U\cup\left\{ \sigma_{j}\right\} $
\end{enumerate}
\item $\tau\leftarrow V\left({\cal F};U\right)/500\log n$

{\color{gray-comment} \# Exploitation phase:}
\item for $j\in\left\{ 1,\dots,n/2\right\} $:

\begin{enumerate}
\item $U\leftarrow U\setminus\left\{ \sigma_{j}\right\} $
\item if $V\left({\cal F}_{W\cup\left\{ \sigma_{n/2+j}\right\} };U\cup W\cup\left\{ \sigma_{n/2+j}\right\} \right)>V\left({\cal F}_{W};U\cup W\right)$\\
{\color{gray-comment} \# does item $\sigma_{n/2+j}$ contribute to
the best feasible set?}

\begin{enumerate}
\item $W\leftarrow W\cup\left\{ \sigma_{n/2+j}\right\} $
\end{enumerate}
{\color{gray-comment} \# if so, add it to $W$}
\item if $\left|W\right|\geq\tau$

\begin{enumerate}
\item return $W$
\end{enumerate}
\end{enumerate}
\end{enumerate}
\end{algorithm}

\subsection{Analysis}

\paragraph{Setup}

Our analysis compares the number of times the value of the current
solution decreases when we forget $\sigma_{j}$ (``bad'' events),
with the number of ``good'' events, where the value of the current
solution increases when we add $\sigma_{n/2+j}$. 

Let $U_{j}$ denote the set $U$ {\em after} we remove the $j$-th
item, i.e. $U_{j}\triangleq\left\{ \sigma_{j+1},\dots,\sigma_{n/2}\right\} $.
We henceforth let $W$ refer to the set returned by the algorithm;
for intermediate values, we let $W_{j}$ denote the set $W$ {\em before}
we consider item $\sigma_{n/2+j}$, i.e. $W_{j}\triangleq W\cap\left\{ \sigma_{n/2+1},\dots,\sigma_{n/2+j-1}\right\} $. 

We want to take a union bound over all choices the algorithm could
have made. Fix any choice of feasible set $\widehat{W}$ of size at
most $\tau$ and a choice $\widehat{\sigma}\mid_{\widehat{W}}$ of
arrival times for those items. For any $j$, let $\widehat{W_{j}}\triangleq\widehat{W}\cap\left\{ \widehat{\sigma}_{n/2+1},\dots,\widehat{\sigma}_{n/2+j-1}\right\} $.
For each $j$, let $B_{j}\left(\widehat{W},\left(\widehat{\sigma}\mid_{\widehat{W}}\right)\right)$
($B$ for ``bad'') be the event that forgetting $\sigma_{j}$ decreases
the value of the current solution, i.e. 
\[
B_{j}\left(\widehat{W},\left(\widehat{\sigma}\mid_{\widehat{W}}\right)\right)\triangleq\begin{cases}
1 & V\left({\cal F}_{\widehat{W}_{j}\cup\left\{ \sigma_{j}\right\} };U_{j}\cup\widehat{W_{j}}\cup\left\{ \sigma_{j}\right\} \right)>V\left({\cal F}_{\widehat{W}_{j}};U_{j}\cup\widehat{W}_{j}\right)\\
0 & \mbox{otherwise}
\end{cases};
\]
similarly, let $G_{j}\left(\widehat{W},\left(\widehat{\sigma}\mid_{\widehat{W}}\right)\right)$
($G$ for ``good'') denote the event that adding $\sigma_{n/2+j}$
increases the value of the current solution,
\[
G_{j}\left(\widehat{W},\left(\widehat{\sigma}\mid_{\widehat{W}}\right)\right)\triangleq\begin{cases}
1 & V\left({\cal F}_{\widehat{W}_{j}\cup\left\{ \sigma_{n/2+j}\right\} };U_{j}\cup\widehat{W_{j}}\cup\left\{ \sigma_{n/2+j}\right\} \right)>V\left({\cal F}_{\widehat{W}_{j}};U_{j}\cup\widehat{W}_{j}\right)\\
0 & \mbox{otherwise}
\end{cases}.
\]
By symmetry, we have that 
\begin{equation}
\Pr\left[B_{j}\left(\widehat{W},\left(\widehat{\sigma}\mid_{\widehat{W}}\right)\right)\right]=\Pr\left[G_{j}\left(\widehat{W},\left(\widehat{\sigma}\mid_{\widehat{W}}\right)\right)\right].\label{eq:B_j=00003DG_j}
\end{equation}
When the choice of $\left(\widehat{W},\left(\widehat{\sigma}\mid_{\widehat{W}}\right)\right)$
is clear from the context, we will simply write $B_{j}$ and $G_{j}$;
we abuse notation and also use $B_{j}$ and $G_{j}$ to denote the
corresponding indicator random variables. Notice also that we do not
require that $\sigma$ agrees with $\widehat{\sigma}$, and the unions
in the definitions of $B_{j}$ and $G_{j}$ may not be disjoint unions. 

We will of course be interested in the special case where $\widehat{W}=W$
and $\widehat{\sigma}$ agrees with $\sigma$. Notice that as long
as we have not reached the target value $\tau$, the algorithm selects
every $j$ for which the good event $G_{j}\left(W,\left(\sigma\mid_{W}\right)\right)$
occurs. Thus, the value of our solution is 
\begin{equation}
\left|W\right|=\min\left\{ \tau,\sum G_{j}\left(W,\left(\sigma\mid_{W}\right)\right)\right\} .\label{eq:|W| =00003D G_j}
\end{equation}
Furthermore, the value of the current solution decreases $\sum B_{j}\left(W,\left(\sigma\mid_{W}\right)\right)$
times, and increases $\left|W\right|$ times. At the beginning, the
value is $V\left({\cal F};U_{0}\right)$ and at the end it is $\left|W\right|$.
Canceling $\left|W\right|$ from both sides of the equation, we get
\begin{equation}
V\left({\cal F};U_{0}\right)=\sum B_{j}\left(W,\left(\sigma\mid_{W}\right)\right).\label{eq:V(U) =00003D B_j}
\end{equation}

Our goal is henceforth to lower bound $\sum G_{j}$ in terms of $\sum B_{j}$
for all choices of $\widehat{W}$ and $\left(\widehat{\sigma}\mid_{\widehat{W}}\right)$
simultaneously. For the special case where $\widehat{W}=W$ and $\left(\widehat{\sigma}\mid_{\widehat{W}}\right)=\left(\sigma\mid_{W}\right)$,
this will in particular imply a lower bound on $\left|W\right|$ in
terms of $V\left({\cal F};U_{0}\right)$.

\paragraph*{Concentration}

We will use the following martingale version of Bennett's inequality.
(The main advantage over the more popular Azuma's inequality is that
this theorem depends on the variance rather than an absolute bound
on the values each variable can take.)
\begin{thm}
\label{thm:martingale-concentration}[e.g. \cite{Freedman1975}] Let
$\left(\Omega,{\cal H},\mu\right)$ be a probability triple, with
${\cal H}_{0}\subseteq{\cal H}_{1}\subseteq\dots$ an increasing sequence
of sub-$\sigma$-fields of ${\cal H}$. For each $k=1,2,\dots$ let
$X_{k}$ be an ${\cal H}_{k}$-measurable random variable satisfying
$\left|X_{k}\right|\le1$ and $\E\left[X_{k}\mid{\cal H}_{k-1}\right]=0$.
Let $S_{n}\triangleq\sum_{k=1}^{n}X_{k}$ and $T_{n}\triangleq\sum_{k=1}^{n}\E\left[X_{k}^{2}\mid{\cal H}_{k-1}\right]$.
Then, for every $a,b>0$,
\[
\Pr\left[\left(S_{n}\geq a\right)\wedge\left(T_{n}\leq b\right)\right]\leq\exp\left(\frac{-a^{2}}{2\left(a+b\right)}\right).
\]
\end{thm}

\paragraph{Main argument}

For the analysis, we reveal the items in a special order: at step
$0$ we reveal the pairs $\left\{ \sigma_{j},\sigma_{n/2+j}\right\} $
unordered, i.e. we reveal all the information up to the last $n/2$
random bits which determine which item arrives at time $j$, and which
at time $n/2+j$. (This information corresponds to ${\cal H}_{0}$
in Theorem \ref{thm:martingale-concentration}.) In fact, those pairs
can be chosen adversarially. Then, we reveal the pairs in reverse
order, beginning with $\left(\sigma_{n/2},\sigma_{n}\right)$. (The
knowledge of the last $k$ pairs corresponds to ${\cal H}_{k}$.)

Fix some choice of $\widehat{W}$ and $\widehat{\sigma}\mid_{\widehat{W}}$.
For $k=1,\dots,n/2$, let $X_{k}\triangleq B_{n/2-k+1}-G_{n/2-k+1}$.
Since we fixed $\widehat{W}_{j}$ and $\left(\widehat{\sigma}\mid_{\widehat{W}_{j}}\right)$
in advance, (\ref{eq:B_j=00003DG_j}) implies that the $X_{k}$'s
satisfy $\E\left[X_{k}\mid{\cal H}_{k-1}\right]=0$ as required in
the premise of Theorem \ref{thm:martingale-concentration}. 

By definition, $S_{n}=\sum\left(B_{j}-G_{j}\right)$. Also, for any
$k$ we have that $\E\left[X_{k}^{2}\mid{\cal H}_{k-1}\right]=1$
if (exactly) one of $\left\{ \sigma_{n/2-k+1},\sigma_{n-k+1}\right\} $
affects the value of the current solution, and $\E\left[X_{k}^{2}\mid{\cal H}_{k-1}\right]=0$
otherwise. Thus 
\begin{gather*}
T_{n}\leq\sum\left(B_{j}+G_{j}\right)\leq V\left({\cal F};\left[n\right]\right)+\sum G_{j},
\end{gather*}
where the second inequality follows from (\ref{eq:V(U) =00003D B_j});
in particular, if $\sum B_{j}\geq\sum G_{j}$, we have that $T_{n}\leq2V\left({\cal F};\left[n\right]\right)$.
Theorem \ref{thm:martingale-concentration} now gives 
\begin{eqnarray*}
\Pr\left[\sum\left(B_{j}-G_{j}\right)\geq V\left({\cal F};\left[n\right]\right)/6\right] & = & \Pr\left[\left(S_{n}\geq V\left({\cal F};\left[n\right]\right)/6\right)\wedge\left(T_{n}\leq2V\left({\cal F};\left[n\right]\right)\right)\right]\\
 & \leq & \exp\left(\frac{-V\left({\cal F};\left[n\right]\right)}{156}\right).
\end{eqnarray*}
Taking union bound over all ${n \choose \tau}^{2}\leq n^{2\tau}$
choices of $\widehat{W}$ and $\widehat{\sigma}\mid_{\widehat{W}}$,
we get that with high probability, for all of those choices simultaneously
$\sum\left(B_{j}-G_{j}\right)\leq V\left({\cal F};\left[n\right]\right)/6$.
In particular, with high probability, 
\[
\sum G_{j}\left(W,\left(\sigma\mid_{W}\right)\right)\geq V\left({\cal F};U_{0}\right)-V\left({\cal F};\left[n\right]\right)/6.
\]
Furthermore, with probability at least $1/2$, 
\[
V\left({\cal F};U_{0}\right)\geq V\left({\cal F};\left[n\right]\right)/2,
\]
in which case our algorithm's solution has value at least $\tau\geq V\left({\cal F};\left[n\right]\right)/1000\log n$.

\section{\label{sec:Non-monotone-feasibility-constra}Non-monotone feasibility
constraints}

Let {\sc Non-monotone Prophet} be the analogous variant of {\sc Downward-Closed Prophet}
when the feasibility constraint is not even guaranteed to be downward-closed.
Since the feasibility constraint is non-monotone, let us explicitly
assume that all item values are non-negative.
\begin{thm}
The competitive ratio for {\sc Non-monotone Prophet} with non-negative
values is $\Theta\left(n\right)$.
\end{thm}
\begin{proof}
Our algorithm simply selects any feasible set that contains the item
with the largest expected contribution. Observe that this achieves
at least $OPT/n$ in expectation.

We now prove an $\Omega\left(n\right)$ lower bound on the competitive
ratio. For every $i\in\left[n/2\right]$, the set $\left\{ i,n/2+i\right\} $
is feasible (and these are the only feasible sets). Also, let $X_{i}=0$
always, and $X_{n/2+i}=\begin{cases}
1 & \mbox{w.p. }1/n\\
0 & \mbox{otherwise}
\end{cases}$. Any online algorithm must commit to some $\left\{ i,n/2+i\right\} $
before observing any of the $X_{n/2+i}$'s, and thus has expected
value of $1/n$. The offline optimum, on the other hand, is $1-\left(1-1/n\right)^{n}\approx1-1/e$.
\end{proof}
For the Secretary Problem, there are several ways to generalize to
non-monotone feasibility constraints. For example, we could assume
that the online algorithm knows in advance the entire set system (with
or without the arrival times, but certainly without the values). Alternatively,
the online algorithm could have an oracle to queries of the form ``is
$S$ a subset of a feasible set?''. The following Theorem holds with
respect to either definition.
\begin{thm}
The competitive ratio for {\sc Non-monotone Secretary} is at most
$n$, and at least $\Omega\left(n/\log^{*}n\right)$.
\end{thm}
\begin{proof}
Our algorithm simply selects the first item and completes any feasible
set that contains it. Since it is a uniformly random item, its value
is at least $OPT/n$.

For the lower bound, we instantiate the Hadamard Code with block length
$n/2$. Recall that it has $n/2$ codewords and the distance between
every two codewords is $n/4$. Let $w_{i}\subseteq\left[n/2\right]$
denote the $i$-th codeword, and consider the set $S_{i}\triangleq w_{i}\cup\left\{ n/2+i\right\} \subset\left[n\right]$;
let the $S_{i}$'s be the feasible sets. Set $X_{n/2+i^{*}}=1$ for
some $i^{*}$, and let all other items have value $0$. Call $w_{i^{*}}$
the {\em good codeword}, and $\left(n/2+i^{*}\right)$ is the index
of the {\em good item}.

In order to achieve value $1$, the online algorithm's selections
must be consistent with the good codeword when the good item arrives.
By symmetry, for each $i$ such that the online algorithm's selection
is consistent with $w_{i}$ when item $\left(n/2+i\right)$ arrives,
it has probability $1/n$ of obtaining value $1$. The expected value
of the online algorithm is therefore proportional to the number of
consistent codewords. Let $k$ be a sufficiently large constant (e.g.
$k=100$). We show that, in expectation, the algorithm is consistent
with $O\left(1\right)$ codewords whose items arrive in time periods
$\left(k\log n,n\right]$, $O\left(1\right)$ codewords whose items
arrive in time periods $\left(k\log\log n,k\log n\right]$, $\left(k\log\log\log n,k\log\log n\right]$,
etc.. In total the algorithm is consistent with $O\left(\log^{*}n\right)$
codewords, so its expected value is $O\left(\log^{*}n/n\right)$,
as opposed to value $1$ obtained by the offline algorithm.

With high probability, every two codewords disagree after $k\log n$
time periods; i.e. for every two sets $S_{i},S_{j}$, at least one
of the first $k\log n$ items is contained in one and not in the other.
Thus after $k\log n$ time periods, the online algorithm must commit
to one subset $S_{i}$, so it only wins if $\left(n/2+i\right)$ is
the good item. In particular, for any choices an online algorithm
may make on the first $k\log n$ items, its expected value from the
items that come in time periods $\left(k\log n,n\right]$ is $O\left(1/n\right)$. 

Similarly, consider the $O\left(\log n\right)$ codewords that correspond
to items that arrive in time periods $\left(k\log\log n,k\log n\right]$.
After $k\log\log n$ codewords, we expect that every two of those
$O\left(\log n\right)$ codewords disagree on some time period. Thus
any online algorithm can obtain expected value at most $O\left(1/n\right)$
from items that arrive in time periods $\left(k\log\log n,k\log n\right]$.
Same for $\left(k\log\log\log n,k\log\log n\right]$, etc. 
\end{proof}
\begin{acknowledgement*}
I thank Tselil Schramm and anonymous reviewers for comments on earlier
drafts. I am also grateful to Sid Barman, Jonathan Hermon, and James
Lee for fascinating discussions about concentration bounds. Most importantly,
I thank Moshe Babaioff and Bobby Kleinberg for insightful suggestions;
in particular, I thank Kleinberg for proposing this problem. I thank
Sahil Singla for pointing out the connection to \cite{AKW14-prophet_limited_info}'s
reduction from Prophet Inequality to order-oblivious Secretary Problems.
\end{acknowledgement*}
\bibliographystyle{alpha}
\bibliography{prophet}

\appendix

\section{\label{sec:List-of-symbols}List of symbols used in Section \ref{sec:Prophet}}
\begin{itemize}
\item $OPT$ - expected (offline) optimum;
\item $W$ - set of items selected by the algorithm;
\item $\ell_{W}\triangleq\max_{i\in W}$ - index of last selected item;
\item ${\cal F}$ - all feasible sets;
\item ${\cal F}_{W}$ - feasible sets that contain $W$;
\item ${\cal F}_{W+j}\triangleq{\cal F}_{W\cup\left\{ j\right\} }$ - feasible
sets that contain $W\cup\left\{ j\right\} $;
\item $X_{i}$ - random value for item $i$;
\item $V\left({\cal F},X_{\left[n\right]}\right)\triangleq\max_{S\in{\cal F}}\sum_{i\in S}X_{i}$
- value of offline optimum solution;
\item $z_{i}$ - observed realization of $X_{i}$;
\item $\tau$ - target value;
\item $\pi\triangleq\Pr\left[V\left({\cal F}_{W},X_{\left[n\right]}\right)>\tau\mid X_{\left[\ell_{W}\right]}=z_{\left[\ell_{W}\right]}\right]$
- target probability;
\item $\pi^{j}\triangleq\Pr\left[V\left({\cal F}_{W+j},X_{\left[n\right]}\right)>\tau\mid X_{\left[j\right]}=\left(z_{\left[\ell_{W}\right]},\underbrace{0\ldots0}_{\ell_{W}+1,\dots,j-1},1\right)\right]$
- probability of reaching target value, when $j$ is the next item
we select;
\item $G=G\left(\tau,\pi,W\right)\triangleq\left\{ j>\ell_{W}\colon\pi^{j}\geq n^{-2}\cdot\pi\right\} \cap\left(\bigcup_{S\in{\cal F}_{W}}S\right)$
- set of good and feasible items;
\item $A=A\left(G\right)\triangleq\left(\bigvee_{j\in G}X_{j}=1\right)$
- the good event where at least one of the good items has a good realization;
\item $\pi_{A}\triangleq\Pr\left[\left(V\left({\cal F}_{W},X_{\left[n\right]}\right)>\tau\right)\mid\left(X_{\left[\ell_{W}\right]}=z_{\left[\ell_{W}\right]}\right)\wedge A\right]$
- probability of reaching target value conditioned on $A$;
\item $A^{j}\triangleq\left(X_{j}=1\bigwedge_{\stackrel{l<j}{l\in G}}X_{l}=0\right)$
- the good event where $j$ is the first good item with a good realization;
\item $\pi_{A^{j}}\triangleq\Pr\left[\left(V\left({\cal F}_{W},X_{\left[n\right]}\right)>\tau\right)\mid\left(X_{\left[\ell_{W}\right]}=z_{\left[\ell_{W}\right]}\right)\wedge A\right]$
- probability of reaching target value conditioned on $A^{j}$;
\item $B^{j}\triangleq\left(X_{j}=0\bigwedge_{\stackrel{l<j}{l\in G}}X_{l}=0\right)$
- the bad event where neither $j$ nor any of the preceding good items
have good realizations.
\end{itemize}

\section{\label{sec:lower_bound}An $\Omega\left(\log n/\log\log n\right)$
lower bound}

For completeness, we briefly sketch an $\Omega\left(\log n/\log\log n\right)$
lower bound due to \cite{BIK07-secretary_original} for the competitive
ratio in both problems. 
\begin{thm*}
[Essentially \cite{BIK07-secretary_original}] There is an $\Omega\left(\log n/\log\log n\right)$
lower bound on the competitive ratios of {\sc Downward-closed Prophet}
and {\sc Downward-Closed Secretary}.
\end{thm*}
\begin{proof}
[Proof sketch] Let the feasible set system be a partition of the
$n$ items into disjoint sets of size $\log n/\log\log n$. Let the
value of each item be $1$ with probability $\log\log n/\log n$,
and $0$ otherwise. With high probability, for at least one of the
feasible sets, all items have value $1$; thus the expected offline
optimum is approximately $\log n/\log\log n$. However, after we select
an item from any feasible set, the expected total value of the remaining
items in this set is less than $1$. Thus no online algorithm can
achieve expected value more than $2$.
\end{proof}

\section{Naive approaches for {\sc Downward-Closed Prophet}\label{sec:Naive-approaches}}

Our algorithm can be described as follows: wait for an item with value
$1$, and select it if that does not decrease the potential function.
In this section we informally discuss a few naive potential functions
that seem to fail. In some sense, our final ``dynamic'' potential
function can be seen as an interpolation between the expectation/median
potential function, and the one based on the probability of beating
a fixed target value.

\paragraph{Recursive potential function}

The optimal online algorithm recursively selects an item iff selecting
it increases the expected value of the optimal algorithm, where ``the
optimal algorithm'' is defined recursively. Unfortunately, this recursive
structure makes this optimal function inconvenient to analyze, and
in particular difficult to compare to the offline optimum solution
benchmark.

\paragraph{Number of maximal sets}

Up to poly-logarithmic factors, we can assume without loss of generality
that all items are identically (but not uniformly) distributed over
$\left\{ 0,1\right\} $, and that all maximal sets have the same size.
Now, every two maximal sets look the same, so one may hope that counting
the number of feasible maximal sets would give a useful potential
function. However, when the set system has a non-trivial structure,
different maximal sets interact very differently with other sets.
In particular, a set that has little or no intersection with other
sets is much more likely to give the best solution, compared to a
set that has a large intersection.

For example, consider the following set system: for the first $2n^{1/3}$
items, every set of size $n^{1/3}$ is a maximal set; as for the remaining
$\left(1-o\left(1\right)\right)n$ items, we divide them into $\left(1-o\left(1\right)\right)n^{5/6}$
blocks of size $n^{1/6}$, and maximal sets are all the unions of
$n^{1/6}$ blocks. Set the probability of each item taking value $1$
to $n^{-1/3}$. Of the first $2n^{1/3}$ items, we expect to see only
a constant number of $1$'s; from the rest we expect to see approximately
$n^{2/3}$ items with value $1$'s, and there must be a feasible set
that contains at least $n^{1/6}$ of them. So taking one of the first
$2n^{1/3}$ items is almost always a bad idea, although there are
a lot more of those sets:
\[
{2n^{1/3} \choose n^{1/3}}=2^{\Omega\left(n^{1/3}\right)}\gg2^{\tilde{O}\left(n^{1/6}\right)}={n^{5/6} \choose n^{1/6}}.
\]
Thus our potential function must take into account the structure of
the set system.

\paragraph{Expectation of the current optimum}

A quantity of interest is the value of an (offline) optimum solution,
among the sets that remain feasible given the items that we have already
selected. This optimum is useful because it easily relates to our
benchmark (the optimum among all sets), it aggregates well both our
current commitment to items and the structure of the set system, and
it exhibits some nice statistical properties such as concentration
(Theorem \ref{thm:ledoux}). The main difficulty is that it is a random
variable, and so it is not clear how to extract a potential function.

The most natural candidate for potential function based on the optimum
is its expected value. If we could show that the expected value of
the optimum decreases by at most $O\left(\log n\right)$ whenever
our algorithm gains $1$ from selecting a reasonable item, we would
be done. It turns out that this is not the case. For example, let
the maximal feasible sets be $\left\{ 1,\dots,n/2\right\} $ and $\left\{ n/2+1,\dots,n\right\} $,
and let each item be uniformly distributed over $\left\{ 0,1\right\} $.
The standard deviation is $\Theta\left(\sqrt{n}\right)$, so when
we commit to one of the sets, the expectation goes down by approximately
$\Theta\left(\sqrt{n}\right)$. The same issue arises if we try to
use the median instead of expectation.

\paragraph{Probability that the current optimum is greater than the global optimum}

Another way to use the optimum-among-currently-feasible-sets random
variable is to consider the probability that it is greater than some
target value, say the global offline optimum, or $\left(1/\log n\right)$
times global optimum. Suppose that we select an item only if selecting
it does not decrease the probability of beating the target value by
a factor more than $n^{2}$; we call such items {\em good}, and
{\em bad} otherwise. Clearly, if we start from a high probability
of beating the target value, we have a high probability of observing
a good item with value $1$: together, all the bad items cannot account
for more than a $\left(1/n\right)$-fraction of the total probability
of beating the target value. 

Suppose that every time we select an item, the probability of beating
the target value goes down by a factor of at most $n^{2}$. After
selecting $w$ items, the probability is still at least $n^{-2w}$.
We can use concentration bounds (such as Theorem \ref{thm:ledoux})
to show that if we now subtract $O\left(w\cdot\log n\right)$ from
the target value, the probability of beating the new target value
becomes high again. This means that as long as the original target
value is greater than $O\left(w\cdot\log n\right)$, there is still
a high probability of observing more $1$'s. In some sense, we can
think of this application of concentration bounds as amortizing the
decrease of expectation (or median) over many iterations.

There is a problem with the argument we sketched in the above two
paragraphs. In the first paragraph, we argued that most of the probability
of beating the target value comes from good items. In the second paragraph,
we argued that as long as we have not selected enough items, we expect
to observe more $1$'s in the future. However these $1$'s do not
necessarily correspond to good items. In particular, when the probability
of beating the target value is very low (we quickly reach exponentially
low probabilities), the event of observing a good item with value
$1$ may account for most of this probability, yet occur with probability
much smaller than $1$. 

To illustrate the problem, recall the $O\left(\log n/\log\log n\right)$
lower bound instance from Appendix \ref{sec:lower_bound}. As soon
as we observe a $1$ on the first item from a feasible subset, we
should select it. Afterward, our probability of observing another
$1$ in the same interval is approximately $(1-1/e)$. Furthermore,
we still have a non-trivial chance (roughly $1/n$) of beating the
global optimum ($\Theta\left(\log n/\log\log n\right)$). However,
we would only maintain a non-zero chance of beating the global offline
optimum if we select one of the first few items - but with probability
$1-o\left(1\right)$ they all have value $0$.
\end{document}